\documentclass[12pt]{article}
\usepackage[margin=1.25in]{geometry}
\usepackage[utf8]{inputenc} 
\usepackage[T1]{fontenc}
\usepackage{url}
\usepackage{ifthen}
\usepackage{cite}
\usepackage[cmex10]{amsmath} 

\usepackage{graphicx}
\usepackage{amssymb}
\usepackage{amsthm}
\usepackage{caption}
\usepackage{subcaption}
\usepackage[ruled,vlined]{algorithm2e}

\usepackage{xcolor}
\usepackage{color}

\def\L{{\cal L}}
\newcommand{\E}{{\cal E}}
\newcommand{\V}{{\cal V}}

\newcommand{\bbR}{{\mathbb R}}

\newtheorem{remark}{Remark}[section]
\newtheorem{proposition}{Proposition}[section]

\begin{document}
\title{Online detection of cascading change-points} 


\author{%
  Rui Zhang, Yao Xie, Rui Yao, Feng Qiu
}


\maketitle

\begin{abstract}
We propose an online detection procedure for cascading failures in the network from sequential data, which can be modeled as multiple correlated change-points happening during a short period. We consider a temporal diffusion network model to capture the temporal dynamic structure of multiple change-points and develop a sequential Shewhart procedure based on the generalized likelihood ratio statistics based on the diffusion network model assuming unknown post-change distribution parameters. We also tackle the computational complexity posed by the unknown propagation. Numerical experiments demonstrate the good performance for detecting cascade failures. 
\end{abstract}


\section{Introduction}
\label{sec:intro}

Cascading failure is a critical problem for power systems. A power system consists of a large number of buses connected by lines, so a failure or anomaly in one component weakens the whole system, makes other components more vulnerable and increases their risk of malfunctioning. Cascading failure is the process caused by the initial failure of one component, which propagates to cause the consecutive failure of other components \cite{7254205}. Cascading failures can often lead to major blackouts, and it is essential to deploy fast detection and effective mitigation.

Online detection of cascading failure is a challenging task.  In some cases, we cannot directly monitor a component's failure but rather use indirect measurements to detect the system's changes and infer the failure. For instance,  \cite{wiltshire2007kalman} uses the difference between the true and the estimated voltage angle, \cite{dobson2010new} uses the area angle to monitor the outage in the power network, and \cite{hines2011estimating} uses auto-correlation in the frequency signal increases as the system is getting close to a critical slowdown and applies this property to detect a failure before it occurs. These studies demonstrate the value of measurements in system monitoring and event detection, but there has been little work on using measurements for the online detection of failures.  

A path forward for modeling cascading failure is that the propagation of failures in a power system can be modeled as diffusion networks \cite{bigdatabook}, which is popular for transmission disease modeling \cite{meirom2017detecting} and
information diffusion in social networks \cite{yang2013mixture}. Prior work \cite{dobson2012estimating,ren2008using} also studied estimating the latent network using multiple cascading cascade events power systems. To perform online detection, we assume that the underlying network and the propagation model are given. But since the propagation is typically unknown since they can be arbitrary and due to anomaly, we have to infer the propagation of the network's failure using real-time measurements.

This paper develops an online change-point detection procedure for power system's cascading failure using multi-dimensional measurements over the networks. We incorporate the cascading failure's characteristic into the detection procedure and model multiple changes caused by cascading failures using a diffusion process over networks \cite{gomezuncovering}. The model captures the property that the risk of component failing increases as more components around it fail. Our change-point detection procedure using the generalized likelihood ratio statistics assuming unknown post-change parameters of the measurements and the true failure time (change-points) at each node.

Closely related works in the literature include the following. Since a failure in the power system can be modeled as a change-point when the distribution of measurements changes, \cite{chen2019compound} examines a general setting with multiple change-points in multiple data streams without utilizing the graph topology of the data streams. Several works consider the propagation of change-points:  \cite{raghavan2010quickest} considers change-point propagation in a line-type network,  \cite{kurt2018multi} assumes that all possible propagation paths are equally likely, both of which proposed a propagation model only depends on the last change-points. However, multiple failure nodes may influence the neighboring nodes, such as outages in power systems; \cite{rigatos2009neural} applies fuzzy neural networks to detect changes in measurements and identify anomaly types. Recent seminal work \cite{zou2018quickest} studies a general change-point detection framework while considering the event propagation dynamics is unknown; in contrast, we consider a specific diffusion model for cascading failures motivated by power systems applications.

\section{Problem Setup} \label{sec:ProSetup} 

Consider a graph $\mathcal G = (\V, \E)$, which is 
formed by a set of nodes $\V=\{1,2,3,\dots, N\}$ and a set of edges $\E\subset \V\times \V$. Here $\V$ corresponds to the set of components in the power network, and $\E$ can be constructed according to the physical network or interaction graph \cite{hines2013dual,qi2014interaction}. Assume $\mathcal G$ is undirected; $X_{i,t}\in\bbR$ is the measurement of $i$th node at time $t$, $t=1,\dots, T$. 

We make the following assumptions about the change points, which correspond to the failure times at each node. Assume the true failure time of the $i$th node is $\tau_i^*\in\bbR^+\cup\{\infty\}$, and $\tau^*_{(i)}$ denotes the corresponding ordered failure time. When $\tau_i^* = \infty$, it means that there is no failure on the $i$th node. Let $\tau^* = (\tau_1^*, \tau_2^*,\dots, \tau_N^*)$ denote the vector of all true failure times, which is unknown.

\subsection{Failure (change-point) propagation model}

Consider the following cascading model for the propagation process of the network's failures. We assume that whenever a failure occurs on a node, it increases the neighboring nodes' tendency to fail. Mathematically, we define the influence of node $i$  on node $j$  as $\alpha_{i,j} >0$. We assume $\alpha_{j,i}$ are known since they can be typically estimated beforehand using historical and simulation data given on the topology of the power grid and power flow. We do not know the distribution for the first failure. After the occurrence of the first failure, the distribution of the subsequent failures is determined by the conditional hazard rate (intensity function) $\lambda_i(t)$:
\begin{equation}\label{hr}
	\lambda_i(t) = \left\{ \begin{array}{ll}
	\sum_{j:(j,i)\in\E, \tau_j^*<t} \alpha_{j,i}, & \tau_{(1)}^*<t\leq \tau_i^*,\\
	0, & \hbox{o.w.}
	\end{array}\right.
\end{equation}
We assume that the failed nodes can only affect the neighboring nodes in the graph. The influence of the failed nodes is constant over time. Therefore, each node's hazard rate before failure is a piece-wise constant, starting at 0 and jumps when the failure affects its neighboring nodes. Figure \ref{diffusion} shows an example of the failure propagation process.

\begin{remark}
	Define the history (filtration) up to time $t$ as \[\mathcal H(t) = \{\tau_i^*\leq t, i =1,\dots,N\}.\] 
	Then, given the $\mathcal H(t)$, the conditional intensity function of the $i$th node is defined as:
	\begin{equation*}
	\lambda_i(t) \triangleq \lim_{\Delta t\rightarrow 0}  \frac{\mathbb P\{\tau_i^*\in[t+\Delta t]|\mathcal H(t), \tau_i^*>t\}}{\Delta t}.
	\end{equation*} The distribution of $\tau^*$ is uniquely defined by the conditional hazard rate \cite{gomezuncovering}.
\end{remark}
\begin{figure}
	\centering
	\includegraphics[width= \linewidth]{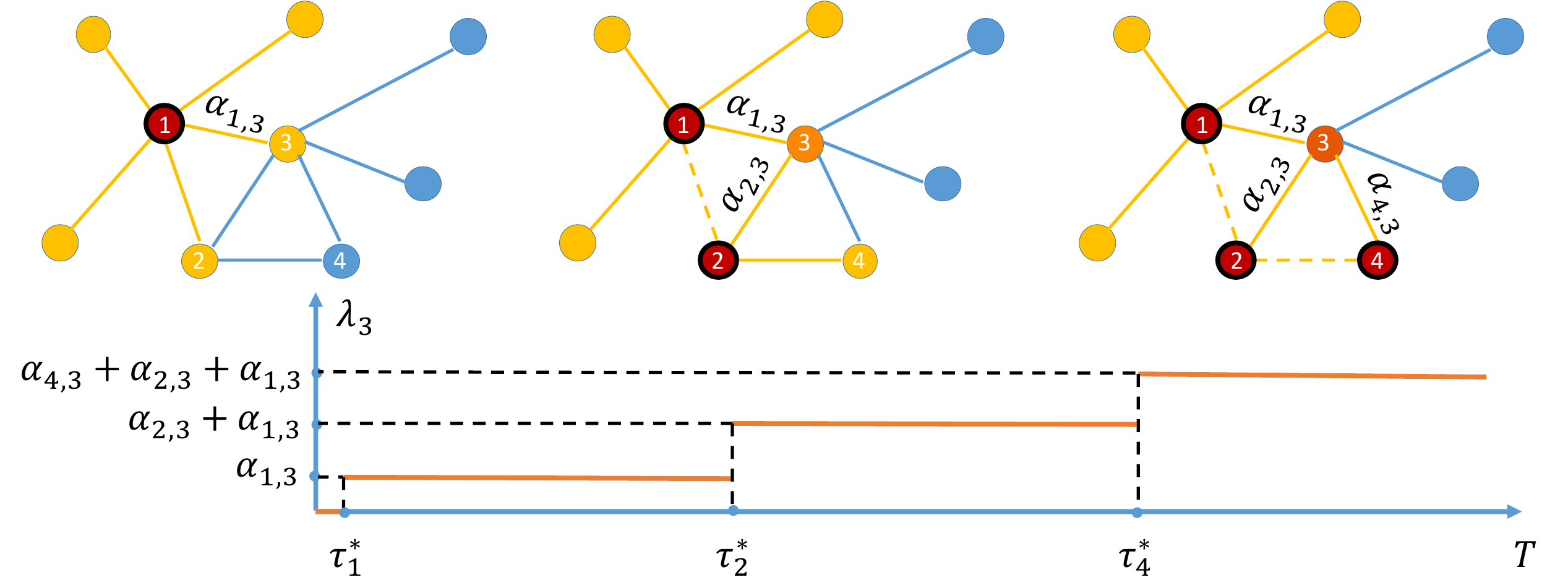}
	\caption{Example of how cascading failure propagates over networks. The failure initiate at node one, then all the neighbors of node one are affected, and node two and node four fail eventually. As the failure propagates, node three is surrounded by more and more failed nodes, and its hazard rate continues to increase. Here, red circles correspond to failed nodes, solid yellow lines are possible paths for failures to diffusion, dashed yellow lines correspond to paths with failed nodes at both ends, yellow circles are nodes affected by failed neighbors.}
	\label{diffusion}
\end{figure}

\subsection{Measurement model}

To simplify the study, we assume that the measurements at each node are independent, conditioned on the failure time. Before a change, they follow an $i.i.d.$ standard normal distribution, and 
after a change they follow an $i.i.d.$ normal distribution with an unknown mean and variance. That is,
\begin{equation}\label{mmodel}
	X_{i,t}\stackrel{\rm i.i.d}{\sim}
	\begin{cases}
		 \,\,\mathcal N(0,1), \,\, &~~ t<\tau_i^*,\\
		 \,\,\mathcal N(\mu_i, \sigma_i^2),\,\, &~~ t\geq\tau_i^*.
	\end{cases}
\end{equation}
Since we can typically use a certain length of data as a warm start to estimate the sample mean and variance of the pre-change distribution, therefore we assume that the pre-change distribution is known and can be standardized.

\subsection{Likelihood function} 

According to the above models, in a time window $[0,T]$, given measurements and failure times, proposition \ref{prop:llh} is the likelihood function.
\begin{proposition}\label{prop:llh}
According to the model defined by Equation (\ref{hr}) and Equation (\ref{mmodel}), the likelihood function for a given $\tau^*$ and $X_{i,t}$ in $[0,T]$ can be expressed as the following:
\begin{align}
	\label{likelihood}
	&f(\tau_i^*, X_{i,t}, \forall i=1,\dots, N, t = 1,\dots,T) \nonumber\\
=& \underbrace{\prod_{i=2}^N f(\tau_{(i)}^*|\tau_{(1)}^*\cdots\tau_{(i-1)}^*)}_{(a)}
\cdot\underbrace{\prod_{i=1}^N\prod_{t=1}^T f(X_{i,t}|\tau_i^*)}_{(b)}
\end{align}
where the term (a) captures the failure propagation model and term (b) captures the measurement model.
\end{proposition}
\noindent
{\it Proof.}
According to model (\ref{mmodel}), given change-points $\tau^*$, $X_{i,t}$ are independent. i.e.
\begin{eqnarray}\label{eq:prop1-1}
   f(X_{i,t}, i = 1,\dots, N, t = 1,\dots, T|\tau^*) = \prod_{i=1}^N\prod_{t=1}^T f(X_{i,t}|\tau^*_i)
\end{eqnarray}
According to \cite{gomezuncovering}, the likelihood of a failure nodes is the product of the survival probability up to the failure time and the  hazard rate at the failure time, i.e. for $\tau^*_i<T$
\begin{eqnarray}\label{eq:sur1}
    f(\tau^*_i|\{\tau^*_1, \dots, \tau^*_N\}\setminus\tau^*_i) = \lambda_i(\tau^*_i)\exp\Big(-\int_{0}^{\tau^*_i}\lambda_i(t)dt\Big).
\end{eqnarray}
For a node which has no failure before time $T$, the likelihood of it is the survival function up to time $T$, i.e. for $\tau^*_i>T$:
\begin{eqnarray}\label{eq:sur2}
    f(\tau^*_i|\{\tau^*_1, \dots, \tau^*_N\}\setminus\tau^*_i) = \exp\Big(-\int_{0}^{T}\lambda_i(t)dt\Big).
\end{eqnarray}
According to the definition of hazard rate (\ref{hr}),  $\lambda_i(t)$ only depends on the change-points before time $t$. Therefore, we have 
\begin{equation}\label{eq:prop1-2}
    f(\tau^*_1,\dots, \tau^*_N) = \prod_{i=2}^N f(\tau^*_{(i)}|\tau^*_{(1)}\dots, \tau^*_{(i-1)}).
\end{equation}
Combine the above and (\ref{eq:prop1-1}), 
we have the following:
\begin{equation}
\begin{split}
         &f(\tau_i^*, X_{i,t}, i=1,\dots, N, t = 1,\dots,T)\\
         &=f(\tau^*_1, \dots, \tau^*_N)f(X_{1,1},\ldots,X_{N,T}|\tau^*_1, \tau^*_2, \dots, \tau^*_N)\\
         &=\prod_{i=2}^N f(\tau^*_{(i)}|\tau^*_{(1)}\dots, \tau^*_{(i-1)})\prod_{i=1}^N\prod_{t=1}^T f(X_{i,t}|\tau^*_i). ~~~~~~~\square
         \end{split}
\end{equation}

Our method can be extended in a more general setting. For the failure propagation model, one can choose different hazard functions. Three models are provided in \cite{gomezuncovering}. In our study, we choose to use the exponential model. For the measurement model, one can select the different distribution, and the measurement can be a high dimension. 

\section{Detection Procedure}\label{sec:testStat}

Consider the following sequential hypothesis test for detecting a dynamic change. An alarm is raised when there are at least $\eta$ change-points. To perform the online detection, at each time instance $T$ we consider the following hypothesis test:
\[
    H_{0,\eta, T}: \tau_{(\eta)}^* >T,\,\,\,
	H_{1,\eta, T}: 0\leq \tau_{(\eta)}^* \leq T.
\]
We consider a Shewhart chart type procedure:  at each time, we evaluate a general likelihood ratio (GLR) statistics over a sliding window. The GLR statistic can handle the mean, and post-change distribution variance are unknown. As shown in (\ref{likelihood}), given the failure time and the measurements, the likelihood can be decoupled into two parts: the likelihood of the failure propagation model and the likelihood of measurement model, respectively. 

\subsection{Log likelihood of failure propagation model}  

Define $\mathcal C(i) = \{j\in\V |(j,i)\in\E\}$ to be the set of the $i$th node's neighbors. Given $\tau$, the log-likelihood function for $[0,T]$  is shown in proposition \ref{prop:llh-f}:
\begin{proposition}\label{prop:llh-f}
Given $T$, a set of failure times $\tau = (\tau_1,\dots, \tau_N)$, graph $\mathcal G$, and parameters $\alpha_{i,j}$ $\forall i, j \in \V$, the log-likelihood function of the failure propagation model is given by
\begin{equation}
\begin{split}
	&\ell_{1,T} = \log f(\tau_1,\tau_2,\dots, \tau_N|\{\alpha_{i,j}\}) \label{llh_f}\\
	& = \sum_{\substack{i:\tau_i\leq T,\\ \tau_i\neq \tau_{(1)}}}\Big\{ \log\left(\sum_{j\in \mathcal C (i)}\alpha_{j,i}\mathbb I(\tau_j<\tau_i)\right) \\
	&- \sum_{j\in\mathcal C(i)}\alpha_{j,i}(\tau_i-\tau_j)^+\Big\} - \sum_{i:\tau_i>T}\sum_{j\in\mathcal C(i)}\alpha_{j,i}(T-\tau_j)^+, 
	\end{split}
\end{equation}
where $(\cdot)^+ = \max(\cdot,0)$, and $\mathbb I(\cdot)$ is the indicator function.
\end{proposition}
\begin{proof}
Combine the model (\ref{hr}), equations (\ref{eq:sur1}) and (\ref{eq:sur2}), the log likelihood of a failure node is, $\forall \tau_i<T$,
\begin{equation}\label{eq:prop2-1}
\begin{split}
&\log f(\tau_i|\{\tau_1, \dots, \tau_N\} \setminus\tau_i)
    = \log\lambda_i(\tau_i) -\int_{0}^{\tau_i}\lambda_i(t)dt\\
    &=\log\Big(\sum_{j\in\mathcal C(i), \tau_j<\tau_i}\alpha_{j,i}\Big) -\sum_{j\in\mathcal C(i), \tau_j<\tau_i}\alpha_{j,i}(\tau_i - \tau_j).
\end{split}    
\end{equation}
The log-likelihood function of a node without failure before time $T$ is, $\forall \tau_i>T$, can be written as
\begin{equation}
\label{eq:prop2-2}
\begin{split}
    &\log f(\tau_i|\{\tau_1, \dots, \tau_N\}\setminus\tau_i) 
    = -\int_{0}^{T}\lambda_i(t)dt\\
    &= - \sum_{j\in\mathcal C(i), \tau_j<T}\alpha_{j,i}(T - \tau_j).
\end{split}
\end{equation}
Combine with equations (\ref{eq:prop1-2}), (\ref{eq:prop2-1}) and (\ref{eq:prop2-2}), we can derive the proposition.
\end{proof}


\subsection{Log likelihood of measurement model} 

Since we assume that the mean and variance of post-change distribution are unknown, we estimate the $\mu_i$s and $\sigma_i$s by maximum likelihood estimation (MLE): $\hat\mu_i$, $\hat\sigma_i$. Therefore the log-likelihood function of measurements of the $i$th node, given $\tau_i$, is:
\begin{align}\label{llh_m}
	\ell_{2,i,T} =&\log f(X_{i,t}, t = 1,\dots, T|\tau_i)\nonumber\\
	=&
	-\sum_{t=1}^{{ T\wedge(\tau_i-1)}}\frac{X_{i,t}^2}{2} - \sum_{t=\tau_i}^T\frac{(X_{i,t} - \hat\mu_i)^2}{2{\hat\sigma^2_i}}\nonumber\\
	&-\frac{T}{2}\log(2\pi)
	-(T-\tau_i+1)^+\log(\hat\sigma_i).
\end{align}
Since we assume that the distribution of measurements at each node is independent given $\tau$, the log likelihood function of all measurements is the summation of the log likelihood function of each node, i.e., $\ell_{2,T} = \sum_{i=1}^N\ell_{2,i,T}$.
%
Therefore, given the failure time $\tau$ and measurements $X_{i,t}$s, the log likelihood at time $T$ is:
\begin{equation}\label{llh}
	\ell_T(\tau, X_{i,t}\, i=1,\dots, N, t=1,\dots, T)
	=\ell_{1,T} +\ell_{2,T}.
\end{equation}Notice that if $\tau_i>T$ for all $i$, Equation (\ref{llh_f}) equals 0 and $\ell_T$ is the sum of log likelihoods of standard normal distribution for the measurements on each node, according to Equation (\ref{llh_m}). 

To perform the hypothesis test between $H_{0,\eta, T}$ and $H_{1,\eta, T}$, we need to search for $\tau$ such that the log-likelihood in (\ref{llh}) is maximized. Define 
\[U(\eta)=\{\tau:\sum_{i=1}^N \mathbb I(\tau_i\leq T)\geq \eta\},\] and 
\[L(\eta)=\{\tau:\sum_{i=1}^N \mathbb I(\tau_i\leq T)\leq \eta-1\}.\] 
We consider a Shewhart type of detection procedure, and we evaluate the test statistics with the data over a sliding window $[T-L+1, T]$, where $L$ is the length of the window. To detect a change for at least $\eta$ change-points, we apply the following GLR test statistics $\forall\eta = 1,\dots, N$:
\begin{align*}
	S_{\eta,T} = \max_{\tau\in U(\eta)}\ell_T(\tau) - \max_{\tau\in L(\eta)} \ell_T(\tau).
\end{align*}
The corresponding stopping time is 
\[\Gamma = \inf\{T>0: S_{\eta,T}> b\},\] 
for some preset threshold $b$.

\section{Computationally Efficient Algorithm}\label{sec:alg}

We would like to implement change detection online and detect the cascading changes as quickly as possible in practice. Thus, we need a low-complexity algorithm and only search for propagation paths with at most $m$ nodes affected by the failure. The computation cost of the maximum likelihood under the alternative hypothesis is high. For instance, for a fully connected graph with $N$ nodes with observations in time horizon $T$, the computation cost is $\mathcal O(T^mN!/((N-m)!m!)$. We aim to develop a computationally efficient algorithm based on a pruning strategy (similar to the ideas in \cite{rigaill2015pruned, padilla2019optimal}). Our proposed algorithm is described as in Algorithm \ref{alg}, which we describe below in more details. 

To reduce the computation cost, we propose a {\it random sampling strategy} as follows. Since the number of possible propagation paths in a fully connected network grows exponentially as the number of nodes increases. Define  $\mathcal F$  as the failure set that contains the failed nodes, and 
\[\mathcal R = \{j\notin\mathcal F: \exists i\in \mathcal F, \alpha_{i,j}>0\},\] as the risk set. Then, we generate the next possible failure points by randomly picking $P$ points in $\mathcal R$ without replacement with probability vector $\mathbf p = (p_i)_{i\in\mathcal R}$, 
where 
\begin{equation}
\label{eq:sampPro}
    p_i = \tilde p_i(\sum_{j\in\mathcal R} \tilde p_j)^{-1}, \quad \tilde p_i = \sum_{j\in\mathcal F} \alpha_{j,i}.
\end{equation}
 With this scheme, we reduce the number of paths to $O(NP^m)$.

To further reduce computational complexity, we also combine the above with a \texttt{Thinning} algorithm by exploiting the monotonicity of $e^{-x}$ involved in the likelihood function, which is summarize in Algorithm \ref{thinning}. 
Consider the likelihood 
\[\L_T =e^{ \ell_{T}}, \L_{2,i,T} = e^{\ell_{2,i,T}}, \L_{2,T} = e^{\ell_{2,T}}, \L_{1,T} = e^{\ell_{1,T}}.\] 
Given a propagation path, we need to compute the  maximum $\L_T(\tau)$, which is the product of $\L_{1,T}$ and $\L_{2,T}$. Define the $q$th percentile of the $i$th node to be $l_{2,i,q}$. Also define a lower bound $l_1$ for $\L_{1,T}$. Instead of maximizing $\L_T(\tau)$ over all the possible choices, we maximize it only in a thinned set 
\[\{\tau:\L_{2,i,T}(\tau_i)\ge l_{2,i,q},\,\forall i = 1,\dots,N\}\cap\{\tau:\L_{1,T}(\tau)\geq l_1\}.\] Specifically, given $\tau_{i}$, we consider 
\[\tau_{j} \in \{\tau_{i}+1, \tau_{i}+2\dots\} \cup \{\tau_{j}:\L_{2,j, T}(\tau_j)\geq l_{2,j,q}\},\] from the smallest to the largest until $\L_{1,T}< l_1$ as shown in Figure \ref{alg1}. The computation cost of this step is $O(h)$, where $h$ depends on the topology of $\mathcal G$, as well as parameters $\alpha_{i,j}$, $l_1$ and $q$. Moreover, we can show that $h\leq [L(1-q)]^m$.

With the above strategies, we can reduce the computation cost to $O(NP^mh)$, which is linear to $N$, the network's size. As shown in the following numerical examples, we can now efficiently compute the test statistics for a 300-bus power system. We combine \textit{random sampling strategy} and \textit{thinning} to reduce the computation cost using a recursion function \texttt{genNext} as shown in Algorithm \ref{alg:genNext}.
\begin{figure}
	\centering
	\includegraphics[width = 0.7\linewidth]{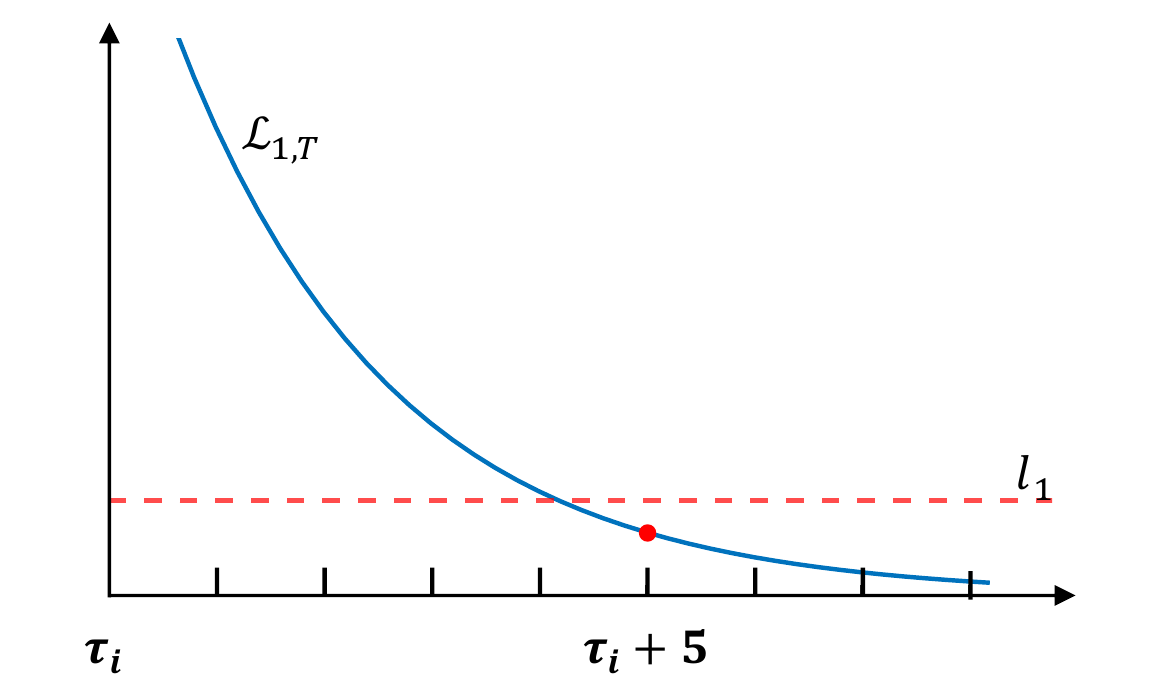}
	\caption{Illustration of the searching strategy for $\tau_j$ given the previous failure point $\tau_i$. When searching $\tau_j$ 
	from $\tau_i+1$, we have $\L_{1,T}(\tau_i+4)>l_1$ and $\L_{1,T}(\tau_i+5)<l_1$; by the monotonicity of $e^{-x}$, we can stop searching at $\tau_i+5$.
}
 
	\label{alg1}
\end{figure}
\begin{figure}
\resizebox{0.88\columnwidth}{!}{
\begin{algorithm}[H]
\SetAlgoLined
{\bf Input}: Data $X_{i,T-L+1}, \ldots X_{i,T}, i = 1, \ldots N$\\
Variables: 
\begin{itemize}
    \item $m$: the maximal number of change-points 
    \item $\ell$: log likelihood given a set of change-points 
    \item $\tau$: a set of change-points
    \item $r$: path of change-points (sorted)
    \item $\ell_{\rm max}$, $\tau_{\rm max}$, $r_{\rm max}$: best log likelihood and parameters
    \item $J$ : sets of potential change-points of each nodes
    \item $j$: current depth of recursion
    \item $q$: percentile to threshold the measurement likelihood
    \item $l_1$: threshold of failure propagation likelihood
\end{itemize}
 1. $j = 1$\\
 2. $J\leftarrow$ compute potential change-points using \texttt{Thinning} ($j$, $m$, $r$, $\tau$, $q$, $l_1$, $\{X_{i,t}\}$)\\
 \For{$x = 1:N$}{
       \For{$t$ in $J(x)$}{
            $p_1 = x$, $\tau_1 = t$\\
            $\ell, r, \tau$ = \texttt{genNext}($j+1, m, r,  \tau$, $\{X_{i,t}\}$)\\
            \If{$\ell>\ell_{\rm max}$}{
			$\ell_{\rm max} = \ell$, $r_{\rm max} = r$, $\tau_{\rm max} = \tau$}
            
            }
       }
 4. \textbf{Return} ($\ell_{\rm max}, r_{\rm max}, \tau_{\rm max}$)\\
 \caption{Cascading change-point detection}
 \label{alg}
\end{algorithm}
}
\end{figure}
\begin{figure}
\resizebox{0.88\columnwidth}{!}{
\begin{algorithm}[H]
\SetAlgoLined
{\bf Input}: $j$, $m$, $r$, $\tau$, Data $\{X_{i,t}\}$ \\
Variable $K$: a set of potential nodes with change-point\\
  \If{$j>m$}{
   $\ell\leftarrow$ compute log-likelihood \\
		return($\ell, r, \tau$)
    }
    $J\leftarrow$ \texttt{Thinning}($j$, $m$, $r$, $\tau$, $q$, $l_1$, Data)\\
	$K\leftarrow$ sample nodes with the probability as (\ref{eq:sampPro})\\ 
	\For{$x$ in $K$}{
	\For{$t$ in $J(x)$}{
		$r_{j} = x$, $\tau_{j} = t$,\\
		$\ell,r, \tau$ = \texttt{genNext}($j+1$, $m$, $r$, $\tau$, $\{X_{i,t}\}$)\\
		\If{$\ell > \ell_{\rm max}$}{
			$\ell_{\rm max} = \ell$,
			$r_{\rm max} = r$,
			$\tau_{\rm max} = \tau$}
	}
	}
	{\bf Return} ($\ell_{\rm max}, r_{\rm max}, \tau_{\rm max}$)
	 \caption{\texttt{genNext} function}
	 \label{alg:genNext}
\end{algorithm}}
\end{figure}

\begin{figure}
\resizebox{0.88\columnwidth}{!}{
\begin{algorithm}[H]
\SetAlgoLined
{\bf Input}: $j$, $m$, $r$, $\tau$, $q$, $l_1$, Data $\{X_{i,t}\}$\\
1. $J= \emptyset $ \\
2. compute $l_{2,x,q}$, $\forall x = 1,\dots, N$ using given Data.\\
\For{$x = 1:N\setminus r$}{
    \For{$t = \tau_j:T$}{
    
    Given $t$ and $\tau$, compute $\mathcal L_{2,x,T}$ and $\mathcal L_{1,T}$.\\
        \If{$\mathcal L_{2,x,T}>l_{2,x,q}$}{
        $J(x) = J(x) \cup t$}
        \If{$\mathcal L_{1,T}<l_1$}{
        exit for    }
    }
}

4.	{\bf Return} ($J$)
	 \caption{\texttt{Thinning} function}
	 \label{thinning}
\end{algorithm}}
\end{figure}

\section{Numerical Examples}\label{sec:numExp}

In this section, we perform several numerical examples to demonstrate our proposed method's performance and compare it to existing methods. We consider two commonly used performance metrics in change-point detection: the average run length (ARL) (a large ARL means a low false alarm rate) and the expected detection delay (EDD). More specifically, for a stopping time $\Gamma$, we define ARL as $\mathbb E_0[\Gamma]$, and we use  $\mathbb E_1[(\Gamma - \tau_{(1)}^*)^+]$ as a measure for EDD (which is a common practice in literature (see, e.g., \cite{xie2013sequential}), where $\mathbb E_i$ denotes the expectation with the probability measure under hypothesis $H_i$. As we increase the threshold, ARL typically increases exponentially, whereas the EDD will increase linearly. A good change-point detection procedure should have a small EDD, given the same ARL.

We consider two case studies: one is to detect the very first change in the network, and the other is to detect the change when there are at least $\eta$ change-points. In Case I, we  compare our methods with generalized likelihood ratio (GLR) and (CuSum), since CuSum is the optimal procedure when the parameter is known \cite{lorden1971procedures}) and GLR is a natural generalization of CuSum when the post-change parameter is unknown \cite{lai2001sequential}. In Case II, we compare our method with the state-of-the-art, including the S-CuSum \cite{zou2018quickest} and a traditional method, Generalized Multi-char CuSum, both of which are suitable for Case II. Below the pre-change distribution is $\mathcal N(0,1)$ and post-change distribution is $\mathcal N(1,1)$.

\vspace{.1in}
\noindent
{\it{Case I}: Detect the first change-point.} In Figure \ref{fig:exp} (left), we show the results in a 300-bus power system (see MATPOWER \cite{zimmerman2010matpower}). In this relatively large system, we apply our algorithm with $L = 100$, $q=0.8$, $P = 1$, $l_1 = e^{-5}$, and $m=5$. 
Our detection statistic can be computed quite efficiently: the average computation time for each time step is less than 3 seconds. Dashed lines are the results when parameters are known. In this scenario, we compare our proposed method with the exact CuSum and two misspecified CuSum ($\mu = 2, 2.5$). Solid lines are the results when parameters are unknown. In this scenario, we compare our proposed method with GLR.
Overall, our method shows the best performance, which is reasonable because our method is not only based on the likelihood ratio but also considers the likelihood of failure propagation.

\vspace{.1in}
\noindent
{\it{Case II}: Detect when there are at least $\eta$ change-points.} Here we compare our proposed method with generalized multi-chart CuSum, and S-CuSum\cite{zou2018quickest}, because these are the most well-known algorithms for tackling such problems. In this experiment, the graph is fully connected with 15 nodes. 
The parameters for the algorithm are $L = 100$, $q = 0.8$, $P=1$, $l_1 = e^{-7}$, and $m = 5$. We set $\eta = 3$. To compute the ARL, we generate data with  $\eta - 1$ affected nodes. 
The result in Figure \ref{fig:exp} (right) shows that our method outperforms both generalized multi-chart CuSum and S-CuSum.

\begin{figure}[t]
	\centering
    	\includegraphics[width = .45 \linewidth]{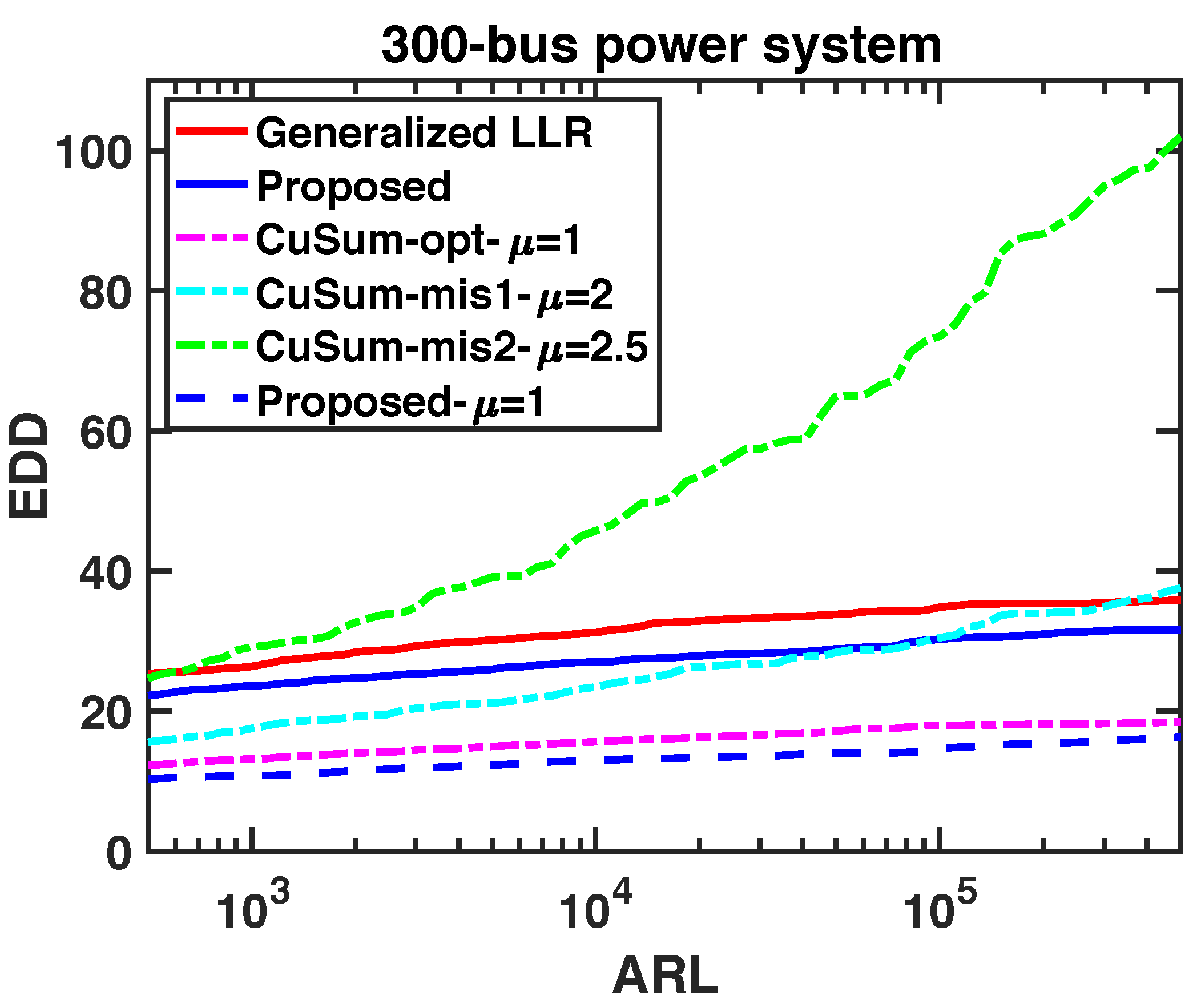}
		\includegraphics[width = .45 \linewidth]{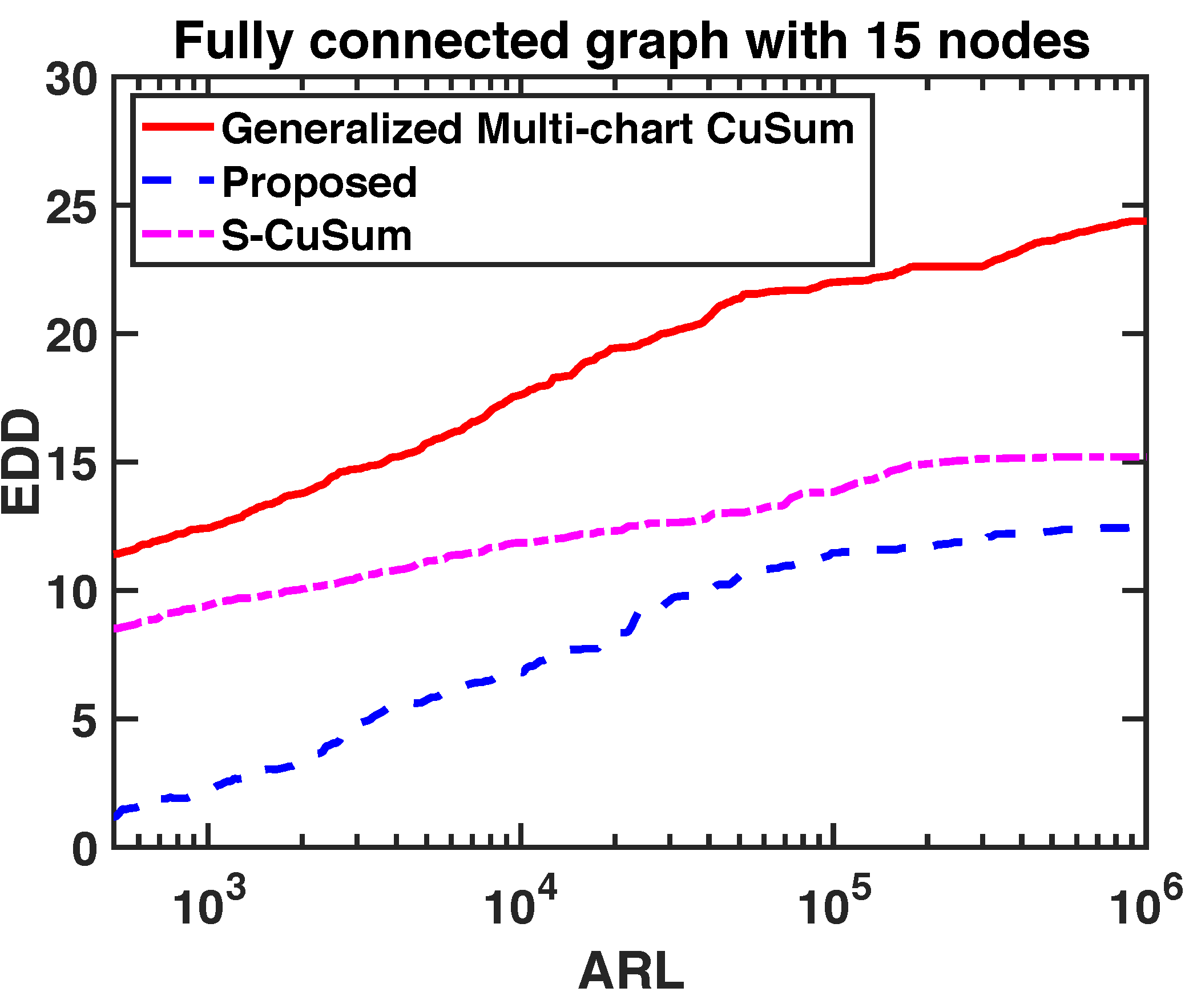}
	\caption{(Left) Comparison of CuSum, generalized likelihood ratio, and the proposed method. (Right) Comparison of generalized multi-chart CuSum, S-CuSum, and the proposed method.}
\label{fig:exp}
	\vspace{-0.2in}
\end{figure}

\section{Conclusion}

We proposed a computationally efficient algorithm to perform the change-point detection by modeling the cascading failure as a temporal diffusion process in a network. Numerical experiments show that our proposed method demonstrates good performance; for  an IEEE 300-bus system, which is considered relatively large,  our results show that the proposed algorithm can scale up to larger systems.

\section*{Acknowledgement}

The work of Rui Zhang and Yao Xie is supported by an NSF CAREER Award CCF-1650913, and NSF CMMI-2015787, DMS-1938106, DMS-1830210.


\newpage


\bibliographystyle{IEEEtran}
\bibliography{References}

\end{document}